\documentclass[conference,a4paper]{IEEEtran}
\usepackage{amsfonts}
\usepackage{amssymb}
\usepackage[cmex10]{amsmath}
\usepackage{amssymb}
\usepackage{graphicx}
\usepackage{graphics}
\usepackage{cite}
\usepackage{subcaption}
\usepackage{dsfont}
\usepackage{url}
\usepackage{flushend}
\newcommand{\F}{\mathbb{F}}
\newcommand{\Z}{\mathbb{Z}}
\newcommand{\C}{\mathbb{C}}

\newcommand{\A}{\mathcal{A}}
\newcommand{\cF}{\mathcal{F}}

\newtheorem{definition}{Definition}

\newtheorem{theorem}{Theorem}
\newtheorem{lemma}{Lemma}

\newtheorem{proof}{Proof}
\newtheorem{remark}{Remark}

\newcommand{\prob}{\mathbf{Pr}}
\newcommand{\defas}{\overset{\text{def}}{=}}

\begin{document}
\title{Probabilistic Shaping and Non-Binary Codes}
%
\author{\IEEEauthorblockN{
Joseph J. Boutros\IEEEauthorrefmark{1}, 
Fanny Jardel\IEEEauthorrefmark{2}, 
and Cyril M\'easson\IEEEauthorrefmark{3}} 
\IEEEauthorblockA{\IEEEauthorrefmark{1}Texas A\&M University, 23874 Doha, Qatar}
\IEEEauthorblockA{\IEEEauthorrefmark{2}Nokia Bell Labs, 70435 Stuttgart, Germany}
\IEEEauthorblockA{\IEEEauthorrefmark{3}Nokia Bell Labs, 91620 Nozay, France\\
boutros@tamu.edu, \{fanny.jardel,cyril.measson\}@nokia.com}}
%
%
\maketitle
%
%
\begin{abstract}
We generalize probabilistic amplitude  shaping (PAS) with binary codes 
\cite{Bocherer2015} to the case of non-binary codes defined over prime finite fields.  
Firstly, we introduce probabilistic shaping via time sharing 
where shaping applies to information symbols only. Then, 
we design circular quadrature amplitude modulations (CQAM)
that allow to directly generalize PAS to prime finite fields
with full shaping.
\end{abstract}
%
%
\section{Introduction}
Shaping refers to methods that adapt the signal distribution to a communication channel for increased transmission efficiency. Shaping is eventually important for optimal information transmissions \cite{Gallager1968} and various solutions starting with non-linear mapping over asymmetric channel models
 towards pragmatic proposals involving shaped QAM signaling have been investigated and/or implemented over the years.

More precisely, building upon early works on, e.g., many-to-one mapping, research efforts from the 70s towards the 90s derive conceptual frameworks and methods to reduce the shaping gap in communication systems. Exploiting the principles of coded modulation, a sequence of works \cite{Calderbank1987,Forney1989a, Forney1989b, Calderbank1990, Fortier1992, Forney1992, Khandani1993, Laroia1994} present operational methods to reduce the shaping gap. Compared to cubic constellations, up to $\frac{\pi \text{e}}6\approx 1.53$dB of {\em shaping gain} is achievable using  well-adapted signaling. Simple methods such as trellis shaping or shell mapping permit to recover a significant fraction of the 1.53dB figure. Examples of applications include the ITU V.34 modem standard recommendations that uses shell mapping to recover 0.8dB.   %
 While several shaping schemes are based on the structural properties of lattices \cite{Forney2000, Urez2005, Pie2016},  more randomized schemes also emerge after the  re-discovery of probabilistic decoding in the 90s. 
With the advent of efficient binary codes, 
different coded modulation schemes were proposed offering 
 flexible and low-complex solutions \cite{Richardson2008}.  
In the 2000s, despite the important development of wireless communications,  the need for advanced shaping methods seems to have remained marginal. From a technological viewpoint, this may have been justified by the high variations of the channel in wireless communication networks. From an academic viewpoint, schemes have been analyzed and match the capacity-achieving distribution of a channel in different theoretical scenarios   \cite{Ling2014, Mondelli2014, Kramer2016}.  %
%
In the last few years, industrial applications of shaping methods have regained interest. This concerns areas where current technologies operate close to fundamental limits. For example, different methods have been experimented in optical communications \cite{Yankov2014, Buchali2015}. 
Hence, because there are already efficient VLSI implementations of contemporary coding schemes that 
 have been proven to asymptotically achieve capacity with constant complexity per information unit \cite{Richardson2008, Lentmaier2010, Arikan2009, Kudekar2013},  it is then natural to  combine them with efficient shaping methods. 

In probabilistic shaping, for linear digital modulations, 
the {\em a priori} probability distribution of modulation points is modified to match a discrete Gaussian-like
distribution, namely the Maxwell-Boltzmann distribution \cite{Kschischang1993}.
The method aims at maximizing the mutual information with respect to the same modulation where
all points are equally likely. 
 For special $2^m$-ASK and $2^m$-QAM constellations with linear binary codes,
this method is equivalent to probabilistic amplitude shaping (PAS) where uniformly-distributed
parity bits are assigned to the sign of a constellation point \cite{Bocherer2015,Schulte2016}.
In this paper, we generalize this method to the non-binary case. 
The goal is to permit the use of efficient non-binary codes
in order to enable low-latency processing 
 (reducing the need for `Turbo'-detection \cite{Richardson2008}).
Also, from an algebraic viewpoint, a characteristic $p>2$ of the finite field $\F_{p^m}$ on which coding is built
leads to new interesting problems such as distribution matching in $\F_{p^m}$ and assigning a constellation
points to elements in $\F_{p^m}$.

In this paper, codes are supposed to be linear and defined over $\F_p$, where $p$ is
an odd prime. Firstly, except for codes with a sparse generator matrix,
we show in Section~\ref{sec_Gallager_Fp} that parity symbols are asymptotically uniformly-distributed over $\F_p$.
This fact is used to 
derive two new methods for probabilistic shaping.
Time sharing is proposed in Section~\ref{sec_time_sharing} where symbols of a $p$-ary code
are mapped into $p$-ASK points. Hence, in time sharing, 
probabilistic shaping is performed only when information symbols are transmitted. 
Full probabilistic shaping is described in Section~\ref{sec_circular_qam}
where circular QAM (CQAM) constellations of size $p^2$ points are introduced. 
This second method assigns a constellation shell to a Maxwell-Boltzmann-distributed information symbol
and then a parity symbol selects a point within that shell.
Numerical results for $p$-ASK-based time sharing and $p^2$-CQAM probabilistic shaping
are shown in Section~\ref{sec_results}.
Similar to the binary case~\cite{Bocherer2015},  
a gap to channel capacity of $0.1$ dB or less is observed for CQAM constellations.
\section{Sum of random variables in a prime field \label{sec_Gallager_Fp}}
Lemma~4.1 in~\cite{Gallager1963} established the expressions of the probability of a sum in $\F_2$. 
We translate this result to a prime field $\F_p=\Z/p\Z$.
Principles of this extension to the non-binary case are implicit from Chapter~5 of~\cite{Gallager1963} with the use of the $z$-transform. Nevertheless, we give here the exact expression of the probability of a sum of prime random symbols 
modulo~$p$. This expression is directly related to Hartmann-Rudolph symbol-by-symbol probabilistic decoding \cite{Hartmann1976} in the special case of a single-parity check code and its generalization to 
characteristic $p$ \cite{Boutros2006}.

\begin{lemma}
Let $p$ be a prime and $\F_p=\{0, 1, \cdots, p-1 \}$ be the associated finite field. 
Consider a sequence $\{s_\ell\}_{\ell=1}^m$ of $m$ independent symbols over $\F_p$ in which the $\ell$-th 
symbol is $\beta \in \F_p$ with probability 
\[
\prob\{s_\ell=\beta\} = q_\ell(\beta). 
\]
Then, for any $k\in\F_p$, the probability that the sum of the $s_\ell$'s equals $k$ is 
\begin{align}
\prob\{\sum_{\ell=1}^m s_\ell =  k\} & = \frac{1+\sum_{i=1}^{p-1}
\prod_{\ell=1}^m \Big( \sum_{\beta=0}^{p-1} q_\ell(\beta) \omega^{i \beta-k+1} \Big)}p,
\end{align}
where $\omega\defas\exp(2\pi\sqrt{-1}/p)$ indicates the $p$-th root of unity.
\end{lemma}
\begin{proof} 
Consider the enumerator function in $t$,
\begin{align}
Q(t)  \defas  \prod_{\ell=1}^m \Big( \sum_{\beta=0}^{p-1} q_\ell(\beta) t^{\beta} \Big) \text{~~mod~~} (t^p-1).
\end{align}
Observe that if this is expanded into a polynomial in $t$ (where degree operations are taken mod $p$), 
the coefficient of $t^k$ is the probability that the sum of $m$ symbols is $k$, which we write 
\begin{equation}
\Pr\{\sum_{\ell=1}^m s_\ell =  k\}=\text{coef}[Q(t),t^k]. 
\end{equation}

Let us also define for any $k\in\{0,1,2,\cdots,p-1\}$ the function 
\begin{equation}
Q_k(t) =  t^{-k} Q(t) \text{~~~mod~~} t^p-1. 
\end{equation}

In an identical manner as for $Q(t)$,  expanding $Q_k(t)$ would enumerate the probabilities of the sum of the $s_\ell$'s. 
Recall that the $p$-th root of unity in the complex plane satisfies $\sum_{i=0}^{p-1}\omega^{k i}=0$ for any $k\in\{1,2,\cdots,p-1\}$. 
Then, for any $k\in\F_p$, we have 
\begin{align}
\prob\{\sum_{\ell=1}^m s_\ell =  k\} & = \frac1p \sum_{i=0}^{p-1} Q_k(\omega^i)
\end{align}
because all but the constant polynomial terms are annihilated from the fact that the roots of  unity sum to zero. It remains to evaluate the expression observing that $\omega^0=1$ to get the results. $\square$
\end{proof}
\begin{remark} This lemma shows that, if a $s_\ell$ is uniformly distributed over $\F_p$, then the sum is also uniformly distributed.
\end{remark}
\begin{remark} For any $\ell$, it is straightforward from convexity arguments that the weighted sum 
$\sum_{\beta=0}^{p-1} q_\ell(\beta) \omega^{i \beta}$ in the complex plane lies inside the unit circle in the strict sense if and only if one of the probability distribution $q_\ell$ is not degenerated in one singular point. Therefore, assuming that the norm tends to be smaller and bounded away from 1, the distribution of the infinite sum tends to be uniform. 
\end{remark}
It remains to summarize this observation in a theorem.
\begin{theorem}
\label{th_uniform_parity}
Let $p$ be a prime and $\F_p=\{0,1,\cdots, p-1\}$ be the associated Galois field. 
Consider a sequence $\{s_\ell\}_{\ell\geq 1}$ of independent random symbols over $\F_p$ 
with respective probability distributions $\{(q_\ell(0),q_\ell(1),\cdots,q_\ell(p-1))\}_\ell\geq 1$ such that $\limsup_{\ell\to\infty}\{\max_{p}(q_\ell(p))\}<1$. 
Then
\begin{align}
\forall k\in\F_p,  & ~~~~\lim_{m\to\infty}\prob\{\sum_{\ell=1}^m s_\ell = k\} = \frac1p.
\end{align}
\end{theorem}
For error-correction over a prime field, this observation is interesting as follows. The limit theorem over $\F_p$ indicates that the non-systematic symbols obtained from a linear encoder associated with a dense generator matrix will tend to have a uniform distribution independently on the input distribution. 

\section{Probabilistic shaping via time sharing over prime fields\label{sec_time_sharing}}
A common mapping between non-binary codes and non-binary modulations is to select
a constellation and a finite field of identical size. Let $p$ be a prime integer,
$p>2$. Consider the set of $p$ points shown in Figure~\ref{fig_p-ASK}, known as $p$-ASK
modulation. This $p$-ASK set $\A=\{-\frac{p-1}{2}, \ldots, -1, 0, 1, \ldots, \frac{p-1}{2} \}$
is isomorphic to the finite field $\F_p$ (a ring isomorphism in $\Z$). Symbols from $\F_p$
are one-to-one mapped into $p$-ASK points. We embed $\F_p$ into $\Z$ such that a symbol $s \in \F_p$
and its corresponding point in $\A$ satisfy $x-s=0 \text{~mod~} p$. 

\begin{figure}[!h]
\centerline{\includegraphics[width=0.8\columnwidth]{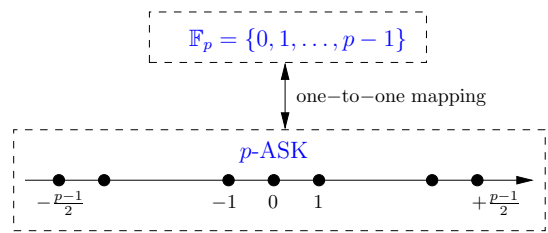}}
\caption{Real $p$-ASK constellation isomorphic to $\F_p$.\label{fig_p-ASK}}
\end{figure}

There are many advantages for such a simple structure where the source is $p$-ary,
the linear code is over $\F_p$, and $p$-ary modulation points are transmitted over the channel.
Firstly, a probabilistic detector needs no conversion between modulation points and code symbols.
A channel likelihood, after normalization, is directly fed as a soft value to the input of a 
probabilistic decoder. Secondly, turbo detection-decoding between the constellation $\A$
and the code $C$ is not required as for binary codes with non-binary modulations \cite{Richardson2008}.\\ 

Consider a systematic linear code $C$ over $\F_p$ where parity symbols satisfy Theorem~\ref{th_uniform_parity}, 
i.e., check nodes used for encoding have a relatively high degree. Many practical error-correcting codes
do satisfy this property, such as LDPC codes
over $\F_p$.
Let $R_c=k/n$ be the coding rate of $C$, where $n$ is the code length and $k$ is the code dimension. 
Assume that information symbols $s_1, s_2, \ldots, s_k$
at the encoder input are identically distributed according to an {\em a priori} 
probability distribution $\{ \pi_i \}_{i=0}^{p-1}$.
Let $P_{MB}(x,\nu) \propto \exp(-\nu |x|^2)$
be a discrete Maxwell-Boltzmann distribution \cite{Kschischang1993} with parameter $\nu \ge 0$.
The {\em a priori} distribution $\{\pi_i\}$ is taken to be
\begin{align}
\label{equ_pi_0} \pi_0 & = P_{MB}(0,\nu) \propto 1, \\
\label{equ_pi_i} \pi_i=\pi_{p-i} &=P_{MB}(i,\nu) \propto \exp(-\nu i^2),
\end{align}
for $i=1\ldots \frac{p-1}{2}$. 
The average energy per point for the $p$-ASK constellation, denoted by $E_s$, is given by
\begin{equation}
E_s=\sum_{x \in \A} P_{MB}(x,\nu) ~|x|^2 ~=~ 2 \sum_{i=1}^{(p-1)/2} \pi_i ~i^2.
\end{equation}

From Theorem~\ref{th_uniform_parity} and (\ref{equ_pi_0})\&(\ref{equ_pi_i}), 
a fraction $R_c$ of transmitted ASK points corresponding to information symbols
are Maxwell-Boltzmann-shaped and a fraction $1-R_c$ of ASK points corresponding to parity symbols
is uniformly distributed in the constellation. We refer to this coding scheme as probabilistic
shaping via time sharing. The target rate should be the average information rate
(expressed in bits per real dimension)
\begin{equation}
\label{equ_Rt_time_sharing}
R_t = R_c \log_2(p) = R_c I(X_s;Y) + (1-R_c) I(X_p;Y),
\end{equation}
where the two random variables $X_s, X_p \in \A$ satisfy $X_s \sim \pi_i$ and $X_p \sim 1/p$.
The random variable $Y$ represents the output of a real additive white Gaussian noise channel, 
where additive noise has variance $\sigma^2=\frac{N_0}{2}$.
For a given target rate $R_t$, the Maxwell-Boltzmann parameter $\nu$ is chosen
such that the signal-to-noise ratio $\gamma=E_s/N_0$ attaining $R_t$ is minimized.
Let $\gamma_{\A}$ be that minimum. We also define two signal-to-noise ratios $\gamma_{cap}$ and $\gamma_{unif}$
such that
\begin{equation}
R_t = \frac{1}{2} \log(1+2\gamma_{cap}),
\end{equation}
and
\begin{equation}
R_t = I(X_p;Y),~~\text{for}~\gamma=\gamma_{unif}.
\end{equation}
Then, the gap to capacity and the shaping gain (expressed in decibels) are respectively given by
$\gamma_{\A}(dB)-\gamma_{cap}(dB)$ and $\gamma_{unif}(dB)-\gamma_{\A}(dB)$.
In this time sharing scheme, probabilistic shaping is made only during a fraction $R_c$
of transmission time. This method is attractive
due to isomorphism between the field $\F_p$ and the $p$-ASK constellation.
From (\ref{equ_Rt_time_sharing}), one may quickly conclude that high coding rate is recommended to approach
full-time  probabilistic shaping. However, at $R_c$ close to $1$, the mutual information
$I(X;Y)$, for $X \in \A$, approaches its asymptote $\log_2(\A)=\log_2(p)$ and the required
signal-to-noise ratio $\gamma_{\A}$ goes far away from $\gamma_{cap}$. 
This is clearly shown in the numerical results presented in Section~\ref{sec_results}.
A method for full probabilistic shaping is proposed in the next section.
\section{Probabilistic shaping via $p^2$-circular QAM over prime fields\label{sec_circular_qam}}
We propose in this section a coded modulation scheme that allows full probabilistic 
shaping of all transmitted symbols with a non-binary linear code over $\F_p$. 
Probabilistic amplitude shaping with binary codes
maps uniformly-distributed parity bits into the sign of an ASK point \cite{Bocherer2015}. 
This sign mapping is not valid with a prime finite field $\F_p$, $p>2$.

The key idea in our new coded modulation is to assign the parity symbol to $p$
modulation points with the same amplitude. This is a direct generalization of the sign mapping
to $p$-ary mapping. The linear $p$-ary code is assumed to be systematic. 
Its information symbols become amplitude labels in the modulation. 
We propose a bi-dimensional constellation with $p^2$ points, 
referred to as {\em $p^2$-circular quadrature amplitude modulation} ($p^2$-CQAM). 
A circle containing CQAM points of the same amplitude will be called a {\em shell}.
The $p^2$-CQAM includes $p$ shells with $p$ points per shell.
Such a bi-dimensional constellation is not unique. Indeed, many ways
do exist to build $p$ shells and populate each shell with $p$ points.
As a consequence, we introduce a figure of merit for a constellation~\cite{Forney1989a} and 
we build a specific $p^2$-CQAM constellation that maximizes this figure of merit.

\begin{definition}
Consider a finite discrete QAM constellation $\A \subset \C$. Assume
that $\sum_{x \in \A} x = 0$.
Let $E_s=\frac{\sum_{x \in \A} |x|^2}{|\A|}$ be the average energy per point,
assuming equiprobable points. Let $d_{Emin}^2(\A)=\min_{x,x' \in \A, x\ne x'} |x-x'|^2$ be
the minimum squared Euclidean distance between the points of $\A$. 
A figure of merit $\cF_M$ for $\A$ is defined by the following expression:
\begin{equation}
\cF_M(\A)=\frac{d_{Emin}^2(\A)}{E_s}\cdot \log_2(|\A|).
\end{equation}
\end{definition}
The $\log_2(|\A|)$ factor is arbitrary, it is used in the above definition to normalize
the squared minimum Euclidean distance by bit energy instead of point energy. This may be
useful when comparing two constellations of different sizes.

Now, we build a $p^2$-CQAM constellation $\A$ that maximizes $\cF_M(\A)$ by populating
the $p$ shells as follows:
\begin{enumerate}
\item For the first shell, draw $p$ uniformly-spaced points on the unit circle. The points are 
$x_{\ell}=\exp(\ell \frac{2\pi}{p} \sqrt{-1})$, for $\ell=0 \ldots p-1$. Here, we impose
the constellation minimum distance to be the distance between two consecutive points of the first shell,
\begin{equation}
d_{Emin}(\A)=2\sin(\frac{\pi}{p}).
\end{equation}

\item Assume that shells $1$ to $i-1$ are already built. 
Let $x_{ip}=\rho_i \exp(\phi_i\sqrt{-1})$ be the first point of the $i$-th shell.
The $p-1$ remaining points on this shell 
are $x_{ip+\ell}=\rho_i \exp((\phi_i+\ell\frac{2\pi}{p})\sqrt{-1})$, $\ell=1, \ldots p-1$.
Let $d_i^2=\min_{\ell=0 \ldots ip-1} |x_{ip}-x_{\ell}|^2$ be the minimum distance
between the first point of the current shell and all previously constructed points.
The radius $\rho_i$ and the phase shift $\phi_i$ are determined by an incremental search:
\begin{itemize}
\item Start with $\rho_i=\rho_{i-1}$ and increment by a step $\Delta_{\rho}$. 
\item At each radius increment, vary $\phi_i$ from $\pi/p$ to $-\pi/p$.
\item Stop incrementing the radius $\rho_i$ and the phase shift $\phi_i$ when $d_i^2 \ge  d_{Emin}^2(\A)$. 
Now, $x_{ip}$ is found. 
\end{itemize}
\item Repeat the second construction step until completing the $p$-th shell of the $p^2$-CQAM constellation.\\
\end{enumerate}
The $p^2$-CQAM obtained with the construction described above has the circular symmetry 
required by PAS over $\F_p$.

\begin{figure*}[!t]
\begin{picture}(400,160)
\put(-12,10){
\put(0,150){\includegraphics[angle=270,width=0.9\columnwidth]{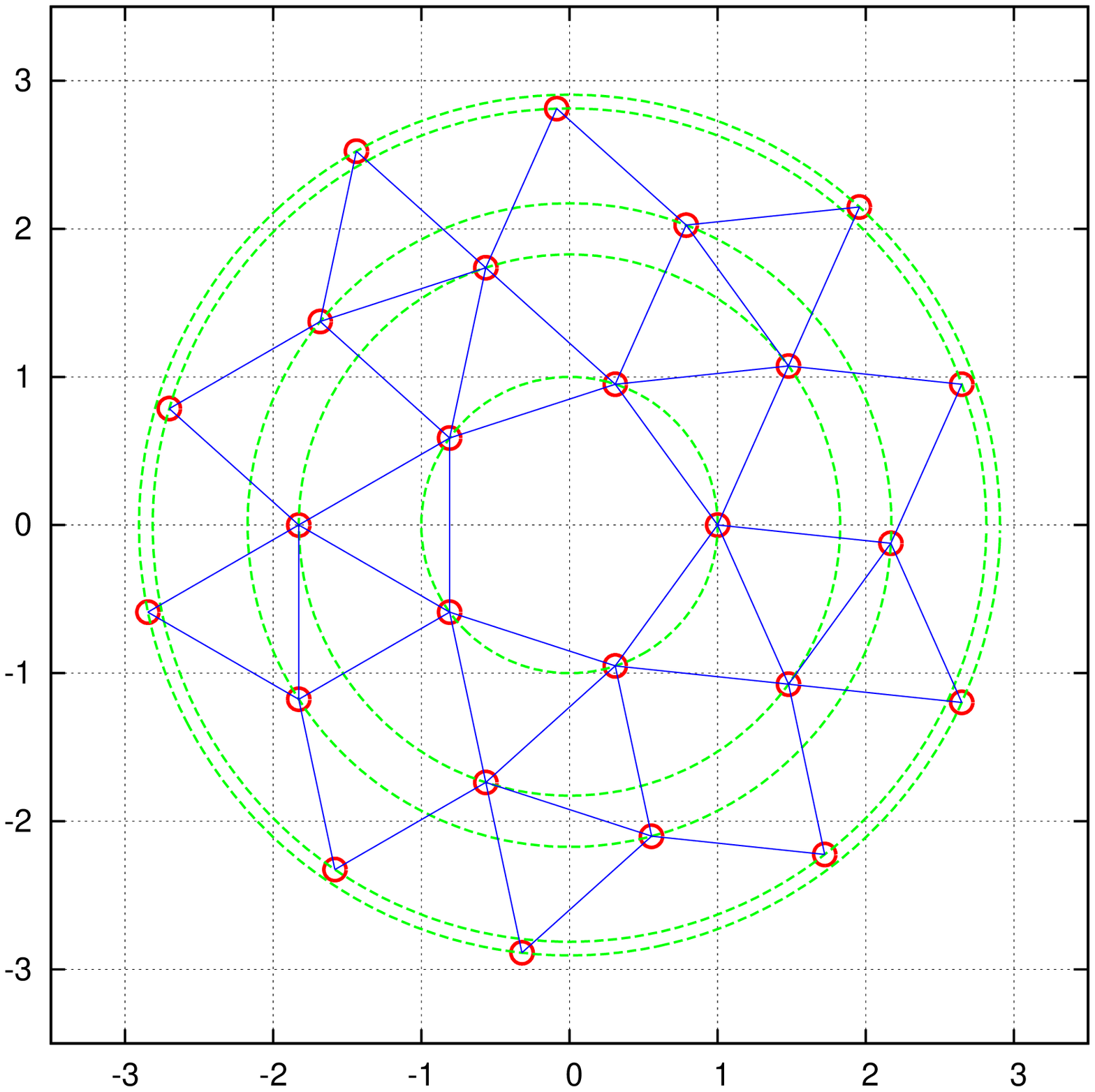}}
\put(160,150){\includegraphics[angle=270,width=0.9\columnwidth]{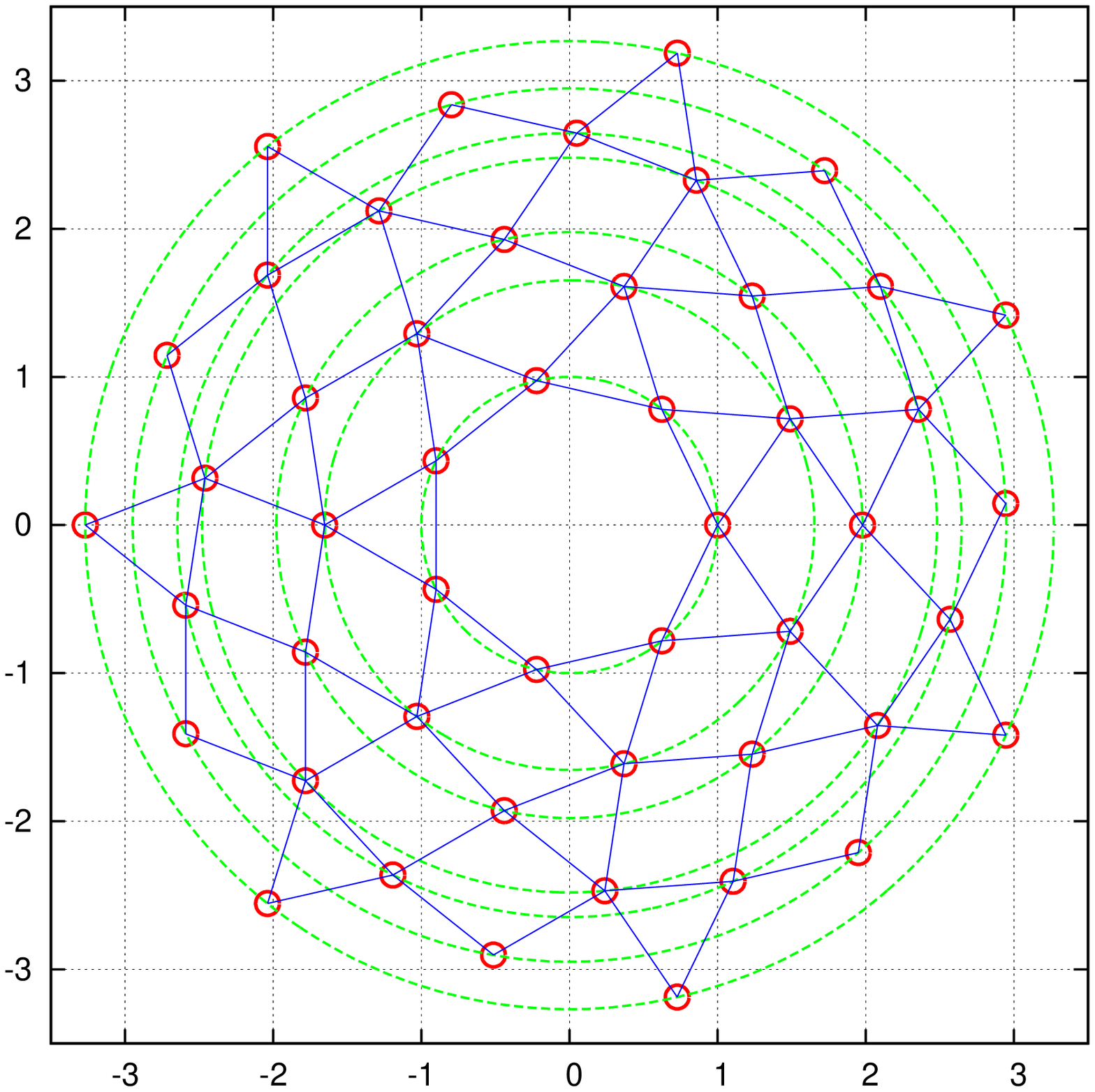}}
\put(320,150){\includegraphics[angle=270,width=0.9\columnwidth]{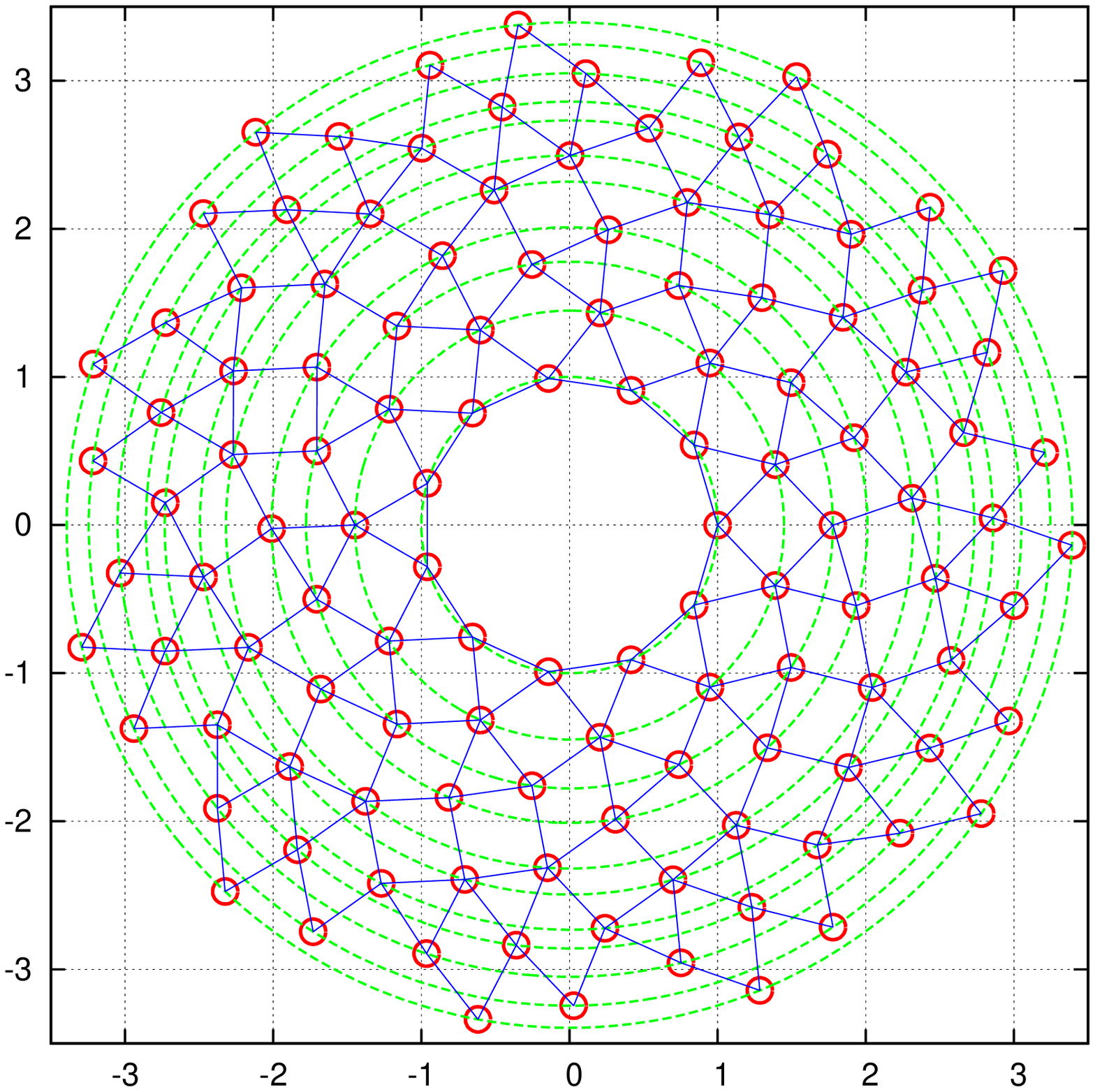}}}
\end{picture}
\caption{Bi-dimensional $p^2$-CQAM constellation for $p=5, 7, 11$ from left to right.\label{fig_p2_CQAM}}
\end{figure*}

Examples of circular QAM modulations for probabilistic amplitude shaping
are shown in Figure~\ref{fig_p2_CQAM}, for $p=5,7,11$ respectively.
Points are drawn as small circles in red. Blue segments connect points located
at minimum Euclidean distance. By the given construction,
the inner radius of the $p^2$-CQAM is $\rho_{in}=1, \forall p$. 
The outer radius $\rho_{out}$ varies slightly with $p$ but 
$\lim_{p \rightarrow \infty} \rho_{out} =\rho_{out}(\infty) \approx 3.6$.
This limit exists
because the sequence $\rho_{out}(p)$ is increasing with $p$
and bounded from above by $1+(p-1)d_{Emin}(\A) \le 1+2\pi$.
The limitation of the Maxwell-Boltzmann probability mass function to amplitudes
between $\rho_{in}=1$ and $\rho_{out}(\infty)$ is a major drawback. 
This short interval $[1,\rho_{out}(\infty)[$ is shifted away from the origin and is not large enough 
to yield a good Gaussian-like discrete distribution. In the next section, the shells radii
are modified to get a wider amplitude range, the $p^2$-CQAM phase shifts are kept invariant.\\

Let $s_1, s_2, \ldots, s_k$ be i.i.d. information symbols with {\em a priori} 
probability distribution $\{ \pi_i \}_{i=0}^{p-1}$, as in the previous section.
Then, for points $x_{ip+\ell} \in \A$, $i,\ell=0 \ldots p-1$, the prior distribution becomes
\begin{equation}
\label{equ_apriori_cqam}
\pi(x_{ip+\ell})=\frac{\pi_i}{p}=\frac{P_{MB}(|x_{ip}|, \nu)}{p}.
\end{equation}
In presence of the above distribution, the signal-to-noise ratio should be defined
with an average energy per point $E_s=\sum_{i=0}^{p-1} \pi_i |x_{ip}|^2$. 
Furthermore, 
the circular symmetry of a $p^2$-CQAM facilitates the numerical evaluation of average mutual information.
The general expression of $I(X;Y)$ with $p^2$ integral terms reduces to $p$ terms only.
The mutual information $I(X;Y)$ is given by
\begin{align*}
\sum_{i=0}^{p-1} \pi_i \int_{y\in \C} p(y|x_{ip}) \log_2\left( 
\frac{p(y|x_{ip})}{\sum_{\ell=0}^{p^2-1} \pi(x_{\ell}) p(y|x_{\ell})}
\right) dy.
\end{align*}
The Maxwell-Boltzmann 
parameter $\nu$ in (\ref{equ_apriori_cqam}) is chosen such that $2R_t=2 R_c \log_2(p)=I(X;Y)$ 
at a minimal value of signal-to-noise ratio $E_s/N_0=\gamma_{\A}$, where $R_t$
is the target rate per real dimension. The gap to capacity is determined by
the difference $\gamma_{\A}-\gamma_{cap}$ with $\gamma_{cap}$ satisfying 
$2R_t=\log_2(1+\gamma_{cap})$. Given the {\em a priori} distribution $\{ \pi_i \}_{i=0}^{p-1}$ 
of symbols in the finite field $\F_p$, the linear code $C[n,k]_p$ and the $p^2$-CQAM constellation
should be combined together as illustrated in Figure~\ref{fig_p2_CQAM_shaping}. 
Suppose that $R_c=1/2$ and say that $s_1 \in \F_p$ is an information symbol (encoder input) 
and $p_1 \in \F_p$ is a parity symbol.
Then, $s_1$ should be be shaped by the distribution matcher (DM)
according to $\{ \pi_i \}_{i=0}^{p-1}$ \cite{Schulte2016}.
The symbol $s_1$ will be used to select a shell in the $p^2$-CQAM and the parity symbol $p_1$
will uniformly select a point on that shell. Similarly, suppose that $R_c=2/3$
and consider four information symbols $s_1, s_2, s_3, s_4$ and two parity symbols $p_1, p_2$.
The DM shall generate $s_1, s_2, s_3$ according to the distribution $\{ \pi_i \}_{i=0}^{p-1}$
and select the shell of three $p^2$-CQAM points.
The symbol $s_4$ is read directly from a uniform i.i.d $p$-ary source. 
Given the shells of three points, uniformly-distributed symbols $s_4, p_1, p_2$ 
constitute the points indices inside those shells. In the general case, 
$n/2$ symbols in $\F_p$ with probability distribution $\{ \pi_i \}$
are read from the DM and mapped into a shell number
for $n/2$ CQAM points. The probabilistic shaping is due to these $n/2$ symbols
On the other hand, $k-n/2$ uniformly-distributed symbols are directly read from the source.
Together with $n-k$ parity symbols, i.e., a total of $n/2$ symbols, uniformly-distributed
symbols in $\F_p$ determine the phase of CQAM points within constellation shells.

\begin{figure}[!h]
\centerline{\includegraphics[width=0.99\columnwidth]{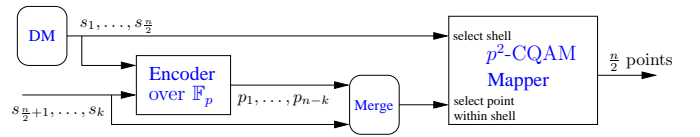}}
\caption{Full probabilistic amplitude shaping with $p^2$-CQAM.\label{fig_p2_CQAM_shaping}}
\end{figure}

Our $p$-ary coded modulation suited for probabilistic shaping assumes that $R_c \ge 1/2$,
i.e., $k \ge n/2$. Coding rates in the interval $[0,1/2[$ are less attractive for probabilistic
amplitude shaping because, for small $R_c$, 
a constellation with equiprobable points already shows
a rate that is too close to channel capacity in terms of signal-to-noise ratio.


\section{Numerical results \label{sec_results}}
Two typical values of $p$ are considered in this section, namely $p=7$ and $p=13$.
The target rate herein is expressed in bits per real dimension for both real
and complex constellations. Gaps and gains are expressed in decibels.\\ 


\begin{table}[h!]
\begin{center}
\begin{tabular}{|c|cccc|c|}
\hline
Size & Coding      & Target & Potential & Gap     & Effective \\
 $p$ & Rate        &  Rate  & Gain      & to cap. & Gain \\ \hline 
7  & $2/3$    & 1.871 & 0.817 & 0.331 & 0.485\\ 
7  & $3/4$    & 2.105 & 0.982 & 0.428 & 0.553\\ 
7  & $4/5$    & 2.245 & 1.105 & 0.546 & {\bf 0.559}\\ 
7  & $17/20$  & 2.386 & 1.283 & 0.753 & 0.530\\ 
7  & $9/10$   & 2.526 & 1.588 & 1.133 & 0.455\\
7  & $19/20$  & 2.666 & 2.232 & 1.916 & 0.316\\ \hline
13 & $2/3$    & 2.466 & 0.997 & 0.346 & 0.651\\ 
13 & $3/4$    & 2.775 & 1.129 & 0.376 & 0.753\\ 
13 & $4/5$    & 2.960 & 1.214 & 0.443 & {\bf 0.771}\\ 
13 & $17/20$  & 3.145 & 1.328 & 0.593 & 0.735\\ 
13 & $9/10$   & 3.330 & 1.549 & 0.915 & 0.633\\ 
13 & $19/20$  & 3.515 & 2.096 & 1.658 & 0.438\\ \hline
\end{tabular}
\end{center}
\caption{\label{tab_time_sharing} Signal-to-noise ratio gain (dB) of probabilistic shaping via time sharing for $7$-ASK and $13$-ASK.}
\end{table}

For probabilistic amplitude shaping via time sharing, signal-to-noise ratio 
gaps and gains are presented in Table~\ref{tab_time_sharing} 
for different values of the coding rate $R_c$. As discussed
in Section~\ref{sec_time_sharing}, the effective gain due to shaping
decreases at very high rate. A coding rate around $4/5$ yields the highest
effective gain. One of our perspectives is to analytically determine the optimal
coding rate (or its range) from (\ref{equ_Rt_time_sharing}). 
Tables~\ref{tab_cqam_gains} includes results for CQAM shaping. 
As suggested in the previous section, CQAM radii are stretched
to improve the range for $P_{MB}(\nu,x)$. Here, radius $\rho_i$ of shell~$i$
is taken to be $1+(\rho_{max}-1) ((i-1)/(p-1))^{\beta}$ where $\rho_{max} > \rho_{out}(\infty)$.
At $R_c=2/3$, optimized parameters are $\rho_{max}=4.8$ and $\beta=0.76$ for $7^2$-CQAM
and $\rho_{max}=6.0$ and $\beta=0.80$ for $13^2$-CQAM.
Square $(p\text{-ASK})^2$ constellations are not valid for full PAS because
they require time sharing, however we added them for comparison purpose.
Shaping with $7^2$-CQAM and $13^2$-CQAM is about $0.1$ dB from the additive white 
Gaussian noise channel capacity.\\



\begin{table}[h!]
\begin{center}
\begin{tabular}{|c|c|c|c|c|}
\hline
Constellation             & Target Rate & Pot. Gain   & Gap   & Full shaping\\ \hline 
$(7\text{-ASK})^2$    & 1.871       & 0.817       & 0.098 & No  \\ 
$7^2$-CQAM            & 1.871       & 0.744       & {\bf 0.101} & Yes \\ \hline
$(13\text{-ASK})^2$   & 2.466       & 0.998       & 0.036 & No  \\ 
$13^2$-CQAM           & 2.466       & 1.092       & {\bf 0.088} & Yes \\ \hline
\end{tabular}
\end{center}
\caption{\label{tab_cqam_gains} Gain (dB) of full probabilistic shaping for circular constellations $7^2$-CQAM and $13^2$-CQAM
at $R_c=2/3$.}
\end{table}

\newcommand{\ieeeit}{{\em IEEE Trans.~Inf.~Theory}}
\newcommand{\ieeecom}{{\em IEEE Trans.~Commun.}}
\newcommand{\ieeesac}{{\em IEEE J.~Sel.~Areas~Commun.}}


\begin{thebibliography}{1}
\bibitem{Bocherer2015}
G.~B{\"o}cherer, F.~Steiner, and P.~Schulte,
``Bandwidth efficient and rate-matched Low-Density Parity-Check coded modulation,''
\ieeecom,  
vol.~63, no.~12, pp.~4651--4665, Dec.~2015.

\bibitem{Schulte2016} 
P.~Schulte and G.~B\"ocherer, 
``Constant composition distribution matching,'' 
\ieeeit, 
vol.~62, no.~1, pp.~430--434, Jan.~2016.

\bibitem{Kschischang1993} 
F.~R.~Kschischang and S.~Pasupathy, 
``Optimal nonuniform signaling for Gaussian channels,'' 
\ieeeit, 
vol.~39, no.~3, pp.~913--929, May 1993.

\bibitem{Gallager1968} 
R.~G.~Gallager, 
{\em Information theory and reliable communication}. 
New York: Wiley, 1968.

\bibitem{Calderbank1987} 
A.~R.~Calderbank and N.~J.~A.~Sloane, 
``New trellis codes based on lattices and cosets,''
\ieeeit,
vol.~33, no.~2, pp.~177--195, Mar.~1987.

\bibitem{Forney1989a} 
G.~D.~Forney and L.-F.~Wei, 
``Multidimensional constellations -- Part I: Introduction, figures of merit, and generalized cross constellations,
\ieeesac, 
vol.~7, no.~6, pp.~877--892, Aug.~1989.

\bibitem{Forney1989b} 
G.~D.~Forney, 
``Multidimensional constellations -- Part II: Voronoi constellations,'' 
\ieeesac, 
vol.~7, no.~6, pp.~941--958, Aug.~1989.

\bibitem{Calderbank1990} 
A.~R.~Calderbank and L.~H.~Ozarow, 
``Non-equiprobable signaling on the Gaussian channel,''
\ieeeit, 
vol.~36, no.~4, pp.~726-740, Jul.~1990.

\bibitem{Fortier1992} 
P.~Fortier, A.~Ruiz, and J.~M.~Cioffi, 
``Multidimensional signal sets through the shell construction for parallel channels,'' 
\ieeecom,
vol.~40, no.~3, pp.~500--512, Mar.~1992.

\bibitem{Forney1992} 
G.~D.~Forney, ``Trellis shaping,'' 
\ieeeit, 
vol.~38, no.~2, pp.~281--300, Mar.~1992.

\bibitem{Khandani1993} 
A.~K.~Khandani and P.~Kabal, 
``Shaping multidimensional signal spaces -- Part 1. Optimum shaping, shell mapping,''
\ieeeit,
vol.~39, no.~6, pp.~1799--1808, Nov.~1993.

\bibitem{Laroia1994} 
R.~Laroia, N.~Farvardin, and S.~A.~Tretter, 
``On optimal shaping of multi-dimensional constellations,''
\ieeeit,
vol.~40, no.~4, pp.~1044--1056, Jul.~1994.

\bibitem{Forney2000} 
G.~D.~Forney, M.~D.~Trott, and S.-Y.~Chung,
``Sphere-bound-achieving coset codes and multilevel coset codes,'' 
\ieeeit, 
vol.~46, no.~3, pp.~820--850, May~2000.

\bibitem{Urez2005} 
U.~Erez, S.~Litsyn, and R.~Zamir,
``Lattices which are good for (almost) everything,'' 
\ieeeit,
vol.~51, no.~10, pp.~3401–-3416, Oct.~2005.

\bibitem{Pie2016} N.~di~Pietro and J.~Boutros, 
``Leech constellations of construction-A lattices,'' 
{\em Sub.~to~IEEE~Trans.~Commun.,} Jan. 2017,
[Online]. Available: http://arxiv.org/pdf/1611.04417v2.pdf.

\bibitem{Gallager1963}
 R.~G.~Gallager, 
{\em Low-Density Parity-Check codes.} 
Cambridge, MA: MIT Press, 1963.

\bibitem{Richardson2008}
T.~Richardson and R.~Urbanke, 
{\em Modern coding theory}.
Cambridge, U.K: Cambridge Univ. Press, 2008.

\bibitem{Ling2014} 
C.~Ling and J.~C.~Belfiore, 
``Achieving AWGN channel capacity with lattice Gaussian coding,'' 
\ieeeit, 
vol.~60, no.~10, pp.~5918--5929, Oct.~2014.

\bibitem{Mondelli2014} 
M.~Mondelli, S.~H.~Hassani, and R.~Urbanke, 
``How to achieve the capacity of asymmetric channels,''
{\em in Proc.~Allerton~Conf.~Commun.~Control~Comput.}, Oct.~2014, pp. 789–796.

\bibitem{Kramer2016} 
G.~Kramer, 
``Probabilistic amplitude shaping applied to fiber-optic communication systems,'' 
{\em in Proc.~Int.~Symp.~on~Turbo~Codes~and~Iterative~Inf.~Proc.,} Oct. 2016.


\bibitem{Lentmaier2010} 
M.~Lentmaier, A.~Sridharan, D.~J.~Costello, and K.~S.~Zigangirov,  
``Iterative decoding threshold analysis for LDPC convolutional codes,''
\ieeeit, 
vol.~56, no.~10, pp.~5274--5289, Oct.~2010.

\bibitem{Arikan2009} 
E.~Arikan, 
``Channel polarization: a method for constructing capacity-achieving codes for symmetric binary-input memoryless channels,''
\ieeeit,
 vol.~57, no.~7, pp.~3051--3073, Jul.~2009.


\bibitem{Kudekar2013} 
S.~Kudekar, T.~Richardson and R.~Urbanke, 
``Spatially coupled ensembles universally achieve capacity under belief propagation,''
\ieeeit, 
vol.~59, no.~12, pp.~7761--7813, Dec.~2013.



\bibitem{Hartmann1976} 
C.R.P. Hartmann and L.D. Rudolph, 
``An optimum symbol-by-symbol decoding rule for linear codes,'' 
\ieeeit,  vol.~22, no.~5, pp.~514-517, Sep.~1976.

\bibitem{Boutros2006} 
J.~Boutros, A.~Ghaith, and Y.~Yuan-Wu, 
``Non-binary adaptive LDPC codes for frequency selective channels: code construction and iterative decoding,''
{\em in Proc. IEEE Inf. Theory Workshop}, 
pp.~184-188, Chengdu, Oct.~2006.




\bibitem{Yankov2014} M.P. Yankov, D. Zibar, K.J. Larsen, L.P.B. Christensen, and S. Forchhammer, ``Constellation Shaping for Fiber-Optic Channels with QAM and High Spectral Efficiency,'' 
{\em IEEE Photon.~Technol.~Lett.}, 
vol.~26, no.~23, pp.~2407--2410, Dec. 2014.

\bibitem{Buchali2015} F. Buchali, G. B{\"o}cherer, W. Idler, L. Schmalen, P. Schulte, F. Steiner, ``Experimental Demonstration of Capacity Increase and Rate-Adaptation by Probabilistically Shaped 64-QAM,'' \emph{in Proc. ECOC}, Aug. 2015.



\end{thebibliography}
\end{document}